\documentclass[a4paper, twocolumn, 10pt]{IEEEtran}
\usepackage[utf8]{inputenc} 
\usepackage[T1]{fontenc}
\usepackage{url}
\usepackage{ifthen}
\usepackage{cite}
\usepackage[cmex10]{amsmath} 
\usepackage{booktabs}
\usepackage{balance}
\usepackage{graphicx}
\usepackage{amssymb}
\usepackage{color}
\usepackage{comment}
\usepackage{arydshln}
\usepackage{algorithm}
\usepackage{algpseudocode}
\usepackage{xcolor}
\usepackage{fixfoot}
\usepackage{enumerate,enumitem}
\usepackage{multicol,multirow}
\usepackage{makecell}
\usepackage{tabularx}
\usepackage{caption}
\usepackage{blkarray}
\usepackage{colortbl}

\usepackage{stmaryrd}

\newcolumntype{Y}{>{\centering\arraybackslash}X} 

\algnewcommand\algorithmicinput{\textbf{Input:}}
\algnewcommand\INPUT{\item[\algorithmicinput]}
\algnewcommand\algorithmicoutput{\textbf{Output:}}
\algnewcommand\OUTPUT{\item[\algorithmicoutput]}
\usepackage{mathtools}

\usepackage{hyperref}
\newcommand\myshade{70} 
\hypersetup{ 
	linkcolor  = red!\myshade!black,
	citecolor  = blue!\myshade!black,
	urlcolor   = blue!\myshade!black,
	colorlinks = true,
}

\usepackage{amsthm}
\theoremstyle{definition}

\newtheorem{proposition}{Proposition}

\newtheorem{theorem}{Theorem}

\newtheorem{remark}{Remark}
\newtheorem{example}{Example}

\newcommand{\bX}{\pmb{X}}
\newcommand{\bY}{\pmb{Y}}
\newcommand{\bZ}{\pmb{Z}}
\newcommand{\bzero}{\pmb{0}}

\newcommand{\vX}{\Vec{\bX}^{(j)}}
\newcommand{\vZ}{\Vec{\bZ}^{(j)}}

\newcommand{\rs}{{\rm{RS}}}
\newcommand{\rank}{{\rm rank}}
\newcommand{\lcm}{{\rm lcm}}

\DeclareFixedFootnote{\mult}{To simplify the notation, we will write $\cdot$ for both $\cdot_\mathbb{U}$ and $\cdot_\mathbb{V}$.}

\newcommand{\hm}[1]{{\color{magenta}(HM: #1)}}

\newcommand{\FF}{\mathbb{F}}

\begin{document}

\title{Verifiable Coded Computation of Multiple Functions}

\author{
\IEEEauthorblockN{
Wilton Kim\IEEEauthorrefmark{1},
\and Stanislav Kruglik\IEEEauthorrefmark{1},
\and Han Mao Kiah\IEEEauthorrefmark{1}
}
\thanks{\IEEEauthorrefmark{1}Authors are with the School of Physical and Mathematical Sciences, Nanyang Technological University, Singapore (email: \{wilt0002, Stanislav.kruglik, hmkiah\}@ntu.edu.sg).}
\thanks{This paper was presented in part at the 2023 IEEE Information Theory Workshop~\cite{Kim2023}. Corresponding author: \textit{Wilton Kim}}
}

\maketitle

\begin{abstract}
We consider the problem of evaluating distinct multivariate polynomials over several massive datasets in a distributed computing system with a single master node and multiple worker nodes. We focus on the general case when each multivariate polynomial is evaluated over its corresponding dataset and propose a generalization of the Lagrange Coded Computing framework (Yu et al. 2019) to perform all computations simultaneously while providing robustness against stragglers who do not
respond in time, adversarial workers who respond with wrong computation and information-theoretic security of dataset against colluding workers. Our scheme introduces a small computation overhead which results in a reduction in download cost and also offers comparable resistance to stragglers over existing solutions.
On top of it, we also propose two verification schemes to detect the presence of adversaries, which leads to incorrect results, without involving additional nodes.
\end{abstract}

\begin{IEEEkeywords}
distributed computing, 
communication efficiency,
verifiability, 
privacy, 
\end{IEEEkeywords}

\section{Introduction}\label{sec::intro}
Due to the enormous size of current datasets, computational operations must be carried out in a distributed manner by outsourcing the workload to external servers~\cite{abadi}. Some of these servers can be \textit{stragglers} (slow-responding servers) (see \cite{strag1, strag2}), \textit{adversarial} (those which respond with wrong computations), or \textit{colluding} (those which communicate with other servers to obtain some information on the datasets) (see \cite{pri1, pri2, pri4, pri5, pri8}). 
Also, we have communication restrictions that limit the scalability of such systems.

Coded distributed computation is an emerging research area that outsources computation to worker nodes in encoded form 
so that computation results are correct despite adversarial behavior.
Such behavior may include providing wrong results, colluding, straggling, or their combination. Polynomial codes were proposed in \cite{strag3} to compute high-dimensional distributed matrix multiplication which tolerates stragglers. In \cite{poly+bgw}, the authors proposed a scheme that combines polynomial codes with Ben-Or, Goldwasser, and Wigderson (BGW) scheme \cite{bgw} to keep the datasets private. 
It was further improved in \cite{lcc}, where the authors proposed \textit{Lagrange Coded Computing} (LCC), which has resiliency against stragglers and adversaries and provides security against colluding workers.
In \cite{mapreduce}, the authors considered the general distributed computing framework, where each server performs as both Master Node and Worker Node, and wants to obtain the computation of several functions $\psi_1,\ldots,\psi_\ell$ on the given dataset $X$. 
However, the distributed computing setup, in this paper, comprises one Master Node that wants to obtain the result of computation on its own data and several Worker Nodes that assist the Master Node (as in \cite{lcc,glcc}). 
We also assume that each Worker Node performs its computational task individually without communication with other Worker Nodes. For a detailed survey of distributed computing, readers can refer to \cite{survey1,survey2}. For the rest of the paper, we call the one who performs the main computational task(s) as the \textit{Master Node} and the other servers as \textit{Worker Nodes}.

Now, to tolerate adversarial Worker Nodes, a typical approach is to collect responses from {\em additional} workers.
This entails an assumption for the maximum allowable adversarial workers, leaving room for possible inaccuracies if this threshold is exceeded. 
Hence, in the second part of the paper, in addition to tolerating stragglers, colluding, and adversarial nodes, we also introduce techniques to verify the correctness of computation results. 
We reiterate that we do not collect results from additional workers.
Instead, we increase the computation load of each worker, albeit marginally, to certify correctness.

Now, verification schemes have a rich history and certain classic verification schemes employ interactive techniques \cite{Babai85,Goldwasser85}.
However, these approaches necessitate multiple rounds of interaction, leading to significant communication overhead. 
To mitigate this issue, the concept of non-interactive verification was introduced in \cite{Goldwasser15}, allowing verification to be performed in one round. 
This breakthrough triggered a flurry of results in verifiable computations \cite{Bitansky12,Catalano15,Gennaro2009,Wang22}. 
The verification methods proposed in this paper fall under the category of non-interactive verification.

Our work is closely related to LCC \cite{lcc}, which evaluates a \textit{single} multivariate polynomial $\psi$ on some datasets in a distributed manner. In \cite{glcc}, the authors proposed a generalization of LCC, entitled {\em Generalized Lagrange Coded Computing} (GLCC), which splits the datasets into several parts so as to define the computational subtasks for the Worker Nodes. 
This gives rise to trade-offs between communication and computation costs and the required number of workers. The Master Node then has the flexibility to decide on how to split the datasets to optimize the performance. 
In both \cite{lcc} and \cite{glcc}, the system comprises one Master Node and many Worker Nodes. The dataset is $\bX = (X_1,\ldots,X_M)$ and the Master Node wants to obtain $\psi(X_1),\ldots,\psi(X_M)$, where $\psi$ is a single polynomial function. For instance, given matrices $X_1,X_2,X_3$, the Master Node wants to obtain $X_1^2,X_2^2,X_3^2$ and here, $\psi(u)=u^2$.

In this paper, {the Master Node wants to evaluate different functions on different elements from the same dataset $\bX$. For example, in the above-mentioned setup} the Master Node wants to obtain $X_1^{10},X_2^{10},X_3^2$, which involve two distinct polynomials $\psi_1(u) = u^{10}$ and $\psi_2(u) = u^2$.  {We discuss briefly some} ways to solve this problem {by modifying existing approaches}:
\begin{itemize}[leftmargin=*]
    \item \textit{Scheme 1}: The Master Node constructs a new polynomial $\psi$ so that it can apply LCC in a single round. For instance, to obtain $X_1^{10},X_2^{10},X_3^2$, the Master Node constructs $\psi(u,v) = u^{10} + v^2$ and views the computations as $\psi(\tilde\bX_1),\psi(\tilde\bX_2),\psi(\tilde\bX_3)$, where $\tilde\bX_1 = (X_1,0),\tilde\bX_2 = (X_2,0), \tilde\bX_3 = (0,X_3)$ and $0$ is the zero matrix of the same dimension as $X_i$.
    \item \textit{Scheme 2}: The Master Node splits the Worker Nodes into $L$ groups, $G_1,\ldots,G_L$ and applies LCC  in each group such that from the group $G_i$, the Master Node obtains all computations from $\psi_i$. For instance, to obtain $X_1^{10},X_2^{10},X_3^2$, the Master Node splits the Worker Nodes into two groups $G_1$ and $G_2$ and applies LCC separately on $G_1$ and $G_2$ such that from $G_1$, the Master Node obtains $X_1^{10},X_2^{10}$ and from $G_2$, the Master Node obtains $X_3^2$. 
    \item \textit{Scheme 3}: The Master Node can apply $L$ rounds of LCC, such that in the $i$-th round, it performs all computations on datasets related with $\psi_i$. For instance, to obtain $X_1^{10},X_2^{10},X_3^2$, the Master Node applies LCC in two rounds.
\end{itemize}

In this paper, we propose a new scheme (defined as Scheme~4) that computes all computations in one round, by modifying the task given to the Worker Nodes. The scheme requires the Worker Nodes to perform slightly more computation and has a slightly worse tolerance to stragglers, but the download cost is significantly lower. 
We elaborate on this in Section~\ref{sec::coded_computation}. 
The key steps in our scheme comprise: 
partitioning the computations into different groups, and 
introducing a polynomial $h$ of degree $K-1$ for all servers to compute.
The polynomial $h$ is designed such that
its evaluations at pre-selected points provide all required computations. 
Hence, to recover all required computations without adversarial nodes, 
the Master Node waits for the first $K$ responses and performs the recovery of $h$.
We employ a similar technique as in \cite{glcc} to construct the function $h$.

Suppose that the Master Node assumes that there are at most $A$ adversaries in the system. 
To tolerate $A$ wrong responses from Worker Nodes, the Master Node needs to wait for an additional $2A$ responses, requiring a total of $K+2A$ responses. 
However, there remains a possibility that the computational results might still be incorrect if the number of adversaries exceeds the designed limit. 
To address this concern, we propose verification schemes that enable us to check (with high probability) for the presence of wrong computational results when we have $K+2A$ responses, among which up to $T$ can be adversarial.

Suppose that the computation results are verified to be wrong; in that case, the Master Node may download responses from two other nodes (totaling $K+2(A+1)$ responses) to obtain and verify the required computations. If incorrect results persist, the Master Node can continue this process as long as the assumed number of adversaries is at most $T$. It is important to note that our verification schemes only detect the presence of wrong computations but do not identify which specific nodes are adversaries.

The rest of the paper is organized as follows. 
In Section~\ref{sec::preliminary}, we formulate the problem and provide an overview of our contributions. Sections~\ref{sec::coded_computation} and~\ref{sec::verification} detail our coded computation schemes and verification schemes, respectively. In Section~\ref{sec::numerical}, we evaluate the performance of the schemes numerically.


\section{Preliminaries}\label{sec::preliminary}
For any positive integer $n$, we denote the set $\{1, 2, \ldots, n\}$ as $[n]$. The finite field of large enough size is denoted as $\FF$. We use $\rs(n, k)$ to represent a Reed-Solomon code of length $n$ and dimension $k$ over the finite field $\FF$. The entropy of a discrete random variable $X$ is denoted as $H(X)$ and the mutual information between two discrete random variables $X$ and $Y$ is denoted as $I(X;Y)$. Additionally, for any matrix $M$, we denote $M^T$ as its transpose.

\subsection{Problem Formulation}\label{sec::problem}

We represent the information at the Master Node (MN) as an $M$-tuple 
$\bX = (X_1,\ldots,X_M)$ with $X_i\in \mathbb{U}$, where $\mathbb{U}$ is a vector space over $\FF$.
For instance, if the elements in the data are square matrices, then $\mathbb{U} = \FF^{n\times n}$. 
Suppose that the Master Node wants to obtain the following $\ell = \sum_{i=1}^L \ell_i$ many computations, 
\begin{align}\label{comp}\Big\{\psi_1\left(\bX_1^{(1)}\right),\ldots,&\psi_1\left(\bX_{\ell_1}^{(1)}\right),\ldots,\notag\\
&\psi_L\left(\bX_1^{(L)}\right),\ldots,\psi_L\left(\bX_{\ell_L}^{(L)}\right)\Big\},\end{align}
where $\bX_i^{(j)}$ contains $s_j$ components from $\bX$ for all $i\in[\ell_j],~j\in[L]$ and $\psi_j:\mathbb{U}^{s_j}\rightarrow{\mathbb{V}}$ is a multivariate polynomial of total degree $D_j$ with
$\mathbb{V}$ being a vector space over $\FF$. 
In other words, the Master Node is interested in $\ell$ computations involving $L$ polynomials, $\psi_1,\ldots,\psi_L$, where each polynomial $\psi_j$ takes $s_j$ elements of $\mathbb{U}$ as input, and produces an element of $\mathbb{V}$ as output. 
For instance, for the case of $\mathbb{U} = \mathbb{V} = \FF^{n\times n}$ and $M=4$, let the dataset be $\bX = (X_1,X_2,X_3,X_4)$. 
Suppose that the Master Node wants to compute $\psi_1(X_1) = X_1^{10}$, $\psi_1(X_2) = X_2^{10}$, $\psi_2(X_3) = X_3^7$ and $\psi_3(X_4) = X_4^2$. In this case, $s_1=s_2=s_3=1$, $L=3$, $\ell_1=2$, $\ell_2=1$, $\ell_3=1$, and $\ell = 4$. 
To obtain required computation results, the Master Node outsources the workload to $d\ge \ell$ Worker Nodes by sending encoded data $\bY_i$ to the $i$-th Worker Node. After receiving $\bY_i$, the $i$-th Worker Node performs some computations on $\bY_i$ and sends its results back to the Master Node. 
The computations are allocated in such a way that the $i$-th Worker Node's response is the evaluation of a polynomial $h$ at point $\alpha_i$. 
We consider the case when there are $S$ stragglers, $A$ adversaries, and $T$ colluding
 nodes (see Figure~\ref{fig:illus}).
\begin{figure}
    \captionsetup{font=small}
    \centering
    \includegraphics[width = 0.4\textwidth]{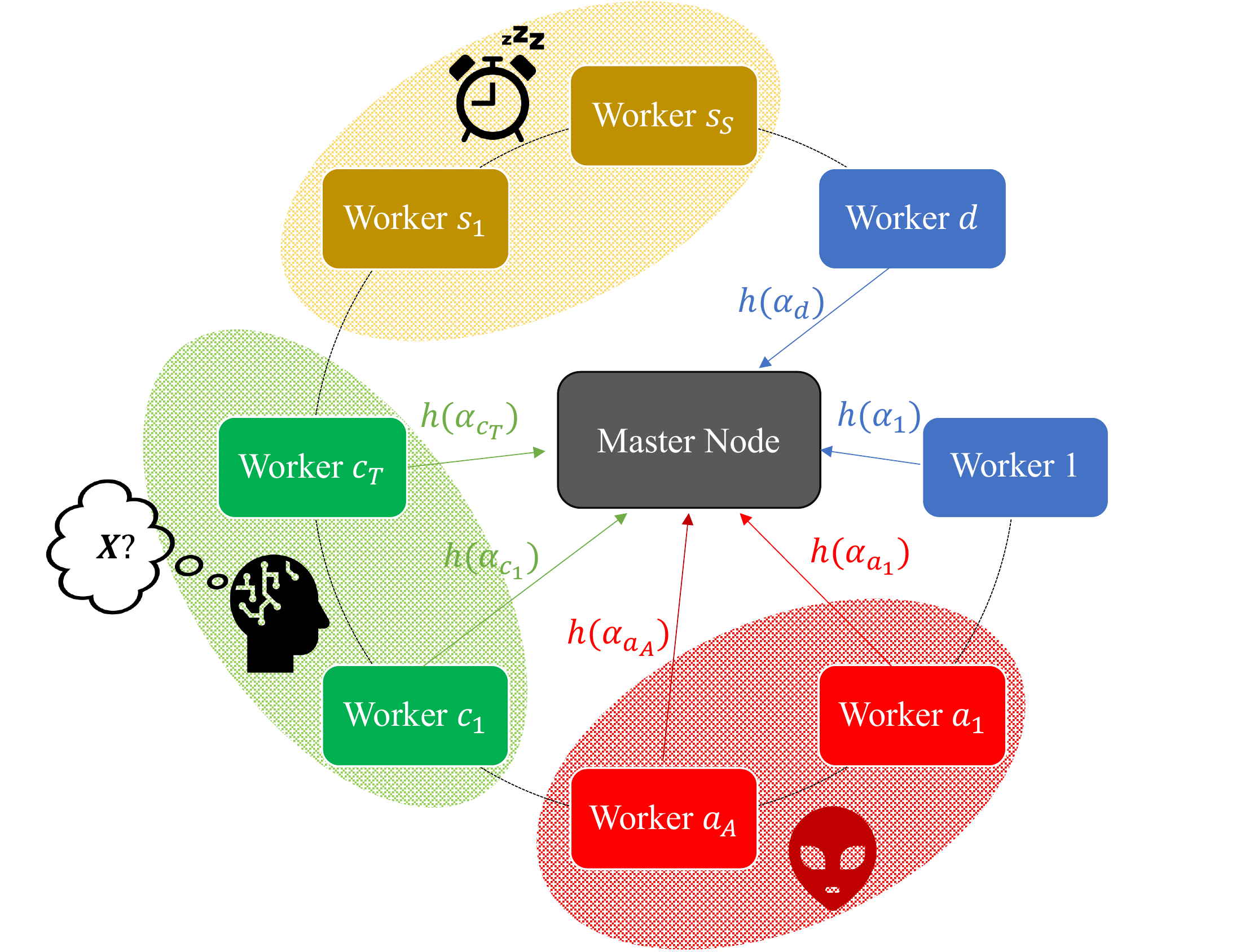}
    \caption{A distributed computing system with one Master Node and $d$ Worker Nodes. To recover computation results, the Master Node waits for the first $K$ Worker Nodes to respond, which might include $A$ adversary responses, and up to $T$ workers may collude to obtain some information on $\pmb{X}$.}
    \label{fig:illus}
\end{figure}
Therefore, we aim to:
\begin{enumerate}[label=(\alph*), leftmargin=*]
    \item Propose a distributed scheme with a small communication cost that satisfies the following constraints:
\begin{itemize}[leftmargin=*]
    \item \textit{$T$-Secure}: Any $T$  colluding nodes are not able to obtain any extra information about the dataset.
    \item \textit{Correctness}: For some $K\le d$, the scheme can correctly recover the required computation from the fastest $K$ responses, even with the existence of $A$ adversarial responses among them. We assume that the identities of adversaries are not known by the Master Node.
\end{itemize}
\item Propose verification techniques to detect the presence of incorrect computations without using additional responses.
\end{enumerate}

\noindent We\,evaluate\,all\,schemes\,with\,the\,following\,performance\,metrics.

\begin{enumerate}[label=(\roman*), leftmargin=*]
    \item Straggler Resistance (SR): The number of stragglers the scheme can tolerate.
    \item Upload Cost (UC): The number of elements in $\mathbb{U}$ the Master Node needs to send.
    \item Download Cost (DC): The number of elements in $\mathbb{V}$ the Master Node needs to download.
    \item Computation in Master Node (MN): The number of $\cdot$ (scalar multiplication) and $\times$ (multiplication of a field element in $\FF$ with an element in $\mathbb{U}$) that the Master Node needs to perform to distribute the computation. 
    We assume that all necessary multiplications of field element $\FF$ are pre-computed.
    \item Computation in Worker Node (WN): The number of multiplications $\cdot$ (scalar multiplication) and $\times$ (multiplication of a field element in $\FF$ with an element in $\mathbb{U}$) that the Worker Node needs to perform to complete the task given by the Master Node. We assume that all necessary multiplications of field elements $\FF$ are pre-computed.
\end{enumerate}

\subsection{Our Contributions}\label{sec::main_result}
In Section~\ref{sec::coded_computation}, we present trivial extensions of 
existing schemes (Schemes~1,~2, and~3) and introduce a new Scheme~4. 
Our proposed Scheme~4 achieves a lower Download Cost compared to other schemes by slightly increasing the workload on Worker Nodes. 
In Section~\ref{sec::verification}, we propose two verification schemes that built upon Scheme 4 to detect the presence of incorrect results among the required computations. 
As before, these verification schemes increase the computation workload of Worker Nodes, 
but do not require additional responses to perform the recovery.  
We provide numerical comparisons of the proposed schemes in Section~\ref{sec::numerical}. 

\section{Coded Computation Schemes}\label{sec::coded_computation}

In this section, we discuss some naive schemes, followed by our proposed scheme to solve the problem mentioned in Section~\ref{sec::problem}, without results verification. First, we consider Scheme 1, where the Master Node rewrites the $\ell$ computations so that they become computations of a single function. Then, we explore another naive scheme, where the Master Node splits the Worker Nodes into $L$ many groups, each performing different computations (Scheme 2). Next, we consider the scheme where the Master Node performs LCC in $L$ rounds (Scheme 3). Afterward, we describe the proposed scheme, where the Master Node performs all computations in one round (Scheme 4).

We note that Schemes~1,~2, and~3 are minor modifications of existing schemes from \cite{lcc}. For the convenience of the reader, we explicitly state their performance in Theorems~\ref{thm:scheme0},~\ref{thm:schemeI}, and \ref{thm:scheme2}.

\begin{remark}
To highlight the main ideas, we only consider schemes based on the LCC framework \cite{lcc}. Nevertheless, it is also possible to apply GLCC \cite{glcc} on top of all our schemes to achieve a more flexible performance in terms of the trade-off between computation and communication costs.
\end{remark}

\subsection{Scheme~1 (First Naive Approach)}
The Master Node forms the multivariate polynomial,

\begin{align}
\Psi(\pmb{u}_1,\pmb{u}_2,\ldots,\pmb{u}_L) = \sum_{i\in [L]} \psi_i(\pmb{u}_i),
\end{align}

where $\pmb{u}_j$ consists of $s_j$ many elements of $\mathbb{U}$. For instance, the computation $\psi_1\left(\bX_1^{(1)}\right)$ can be expressed as $\Psi\left(\bX_1^{(1)},\bzero,\ldots,\bzero\right)$. Hence, the Master Node can reformulate the problem into obtaining the computations
\begin{align}\label{eq:scheme0}
    &\Bigg\{\Psi\left(\tilde\bX_1^{(1)}\right),\ldots,\Psi\left(\tilde\bX_{\ell_1}^{(1)}\right),\ldots, \notag \\
    &\hspace{26mm} \Psi\left(\tilde\bX_1^{(L)}\right),\ldots, 
    \Psi\left(\tilde\bX_{\ell_L}^{(L)}\right)\Bigg\},
\end{align}
where $\tilde\bX_i^{(j)}$ is formed so that its $j$-th component is $\bX_i^{(j)}$ and other components are zero. 
Clearly, $\Psi$ is a polynomial of degree $\max_j\{D_j\}$. Therefore, the Master Node can apply LCC to obtain \eqref{eq:scheme0}. 
Specifically, the Master Node constructs a sharing polynomial $f$, with distinct evaluation points from $\pmb{\beta} \triangleq \bigcup_{j\in[L]}\pmb{\beta}_j\cup\left\{\beta^{(R)}_i:i\in[T]\right\}$, where $\pmb{\beta}_j \triangleq\left\{\beta_i^{(D,j)}\in\FF:i\in[\ell_j]\right\}$, such that
\begin{align}
    \begin{cases}
        f\left(\beta_i^{(D,j)}\right) = \tilde\bX_i^{(j)}&\text{for all }j\in[L],i\in[\ell_j],\\
        f\left(\beta_i^{(R)}\right) = \tilde\bZ_i&\text{for all }i\in[T].
    \end{cases}
\end{align}
Here, the elements in each $\tilde\bZ_i$ are independently and uniformly chosen at random. 
Note that the polynomial $f$ interpolates $\left(\sum_{j\in[L]} \ell_j\right)+T = \ell + T$ points, 
and hence, $f$ is a polynomial of degree $\ell + T-1$. The Master Node assigns a unique evaluation point to each Worker Node from $\pmb{\alpha} = \{\alpha_i\in\FF:i\in [d]\}$ (here, $\pmb{\alpha}\cap\pmb{\beta}=\emptyset$) and sends $f(\alpha_i)$ to the $i$-th Worker Node. The $i$-th Worker Node proceeds to compute \begin{align}\label{eq::workernode_h_scheme1}
   h(\alpha_i) = \Psi(f(\alpha_i)),
\end{align} and sends it back to the Master Node. The Master Node expects to obtain
\begin{align}\label{resp::scheme1}
    (h(\alpha_1),\ldots,h(\alpha_d)),
\end{align}
and this is a codeword of an $\rs(d,K)$ code with $K = \max_j\{D_j\}(\ell + T - 1) + 1$.  
Note that recovering $h$ gives us all required computations.
The values of performance metrics are formulated in the theorem below.
\begin{theorem}[Scheme 1]\label{thm:scheme0}
Fix $A$ and $T$, and set $K_1\triangleq \max_j\{D_j\}(\ell + T - 1) + 1$.
Further, choose $d\ge K_1 +2A$.
 
Suppose that there are $d$ Worker Nodes, of which at most $A$ are adversarial and at most $T$ are colluding.
Then Scheme~1 is $T$-secure, correct, and achieves the following metrics.
\begin{itemize}[leftmargin=*]
    \item Straggler Resistance: $d-K_1-2A$.
    \item Upload Cost: $d\sum_{j\in[L]} s_j$.
    \item Download Cost: $K_1+2A
    $.
    \item Computation\,in\,MN: $(\ell+T)\sum_{j\in[L]}s_j$ multiplications of $\cdot$'s.
    \item Computation\,in\,WN: no multiplications of $\cdot$'s and the number of multiplications $\times$'s in $\Psi$ to compute.
\end{itemize}
\end{theorem}
\begin{proof}
    Since $K-1$ corresponds to the degree of the polynomial $h$ defined in~\eqref{eq::workernode_h_scheme1} and
    Worker Node $i$ computes $h(\alpha_i)$, the $d$-tuple of computations \eqref{resp::scheme1} can be viewed as a codeword of an $\rs(d,K_1)$ code. Hence, with $K_1+2A$ responses, the Master Node can correctly recover the polynomial $h$, even in the presence of $A$ adversarial responses (see for example~\cite[Ch.~6]{assympdecode}).
    
    \noindent \textit{Stragglers Resistance.} Since we require $K_1+2A$ responses, we can tolerate $d-K_1-2A$ stragglers.
    
    \noindent \textit{Upload and Download Cost.} The Master Node sends an encoded data which consists of $\sum_{j\in[L]} s_j$ elements of $\mathbb{U}$ to all Worker Nodes. So, the Upload Cost is $d\sum_{j\in[L]}s_j$ elements of $\mathbb{U}$. To do the recovery, the Master Node downloads $K_1+2A$ elements of $\mathbb{V}$.

    \noindent \textit{Computation in MN.} The Master Node computes the values of a sharing polynomial at some evaluation point, which interpolates $\ell$ points of dimension $\sum_{j\in[L]}s_j$ and $T$ random points of the same dimension. Assuming that all multiplications of field elements are pre-computed, the Master Node needs to perform $(\ell+T)\sum_{j\in[L]} s_j$ scalar multiplications for each node.

    \noindent \textit{Computation in WN.} The Worker Nodes apply $\Psi$ to their received shared information. So, the number of $\times$ performed by the Worker Nodes is equal to the number of $\times $ in $\Psi$.

    \noindent \textit{Security.} 
    Since the security proofs of all Schemes~1--4 are similar, we only provide a detailed proof for Theorem~\ref{theorem:1} (see Section~\ref{sec:scheme3}).
\end{proof}


As mentioned before, there are two other naive techniques. 
One is to split the Worker Nodes into $L$ different groups, $G_1,\ldots,G_L$, such that the computations for Worker Nodes in $G_j$ only involve $\psi_j$. Then for each group, we apply LCC separately. 
Another one is to apply LCC in $L$ rounds, so that in round $j$, we obtain computations involving $\psi_j$. 
Henceforth, for these schemes, the Master Node constructs $L$ \textit{sharing} polynomials, by considering a set of distinct evaluation points $\pmb{\beta} \triangleq \bigcup_{j\in[L]}\pmb{\beta}_j\cup\left\{\beta^{(R)}_i:i\in[T]\right\}$ where $\pmb{\beta}_j \triangleq\left\{\beta_i^{(D,j)}\in\FF:i\in[\ell_j]\right\}$. The points are chosen such that, for all $j\in [L]$, we have $f_j: \FF\rightarrow \mathbb{U}^{s_j}$ and 
\begin{alignat}{4}\label{eq:sharingpoly}
&f_j\left(\beta_1^{(D,j)}\right) &&= \bX_1^{(j)}&&,\ldots, f_j\left(\beta_{\ell_j}^{(D,j)}\right) &&= \bX_{\ell_j}^{(j)},\notag\\
&f_j\left(\beta_1^{(R)}\right) &&= \bZ_1^{(j)}&&,\ldots,f_j\left(\beta_T^{(R)}\right) &&= \bZ_T^{(j)}.
\end{alignat}
Here, the elements in each $\bZ_i^{(j)}$ are independently and uniformly chosen at random.
It is clear that $f_j$ is a polynomial of degree $\ell_j+T-1$. 
The Master Node assigns each Worker Node with a unique evaluation point from set $\pmb{\alpha} = \{\alpha_i\in\FF:i\in[d]\}$ such that $\pmb{\alpha}\cap \pmb{\beta} = \emptyset$. 
In Sections~\ref{sec::scheme2} and~\ref{sec:scheme2}, we study certain schemes that use these sharing polynomials in a straightforward manner.
Later, in Section~\ref{sec:scheme3}, we design a method to synthesize these sharing polynomials and obtain a scheme with better performance.

\subsection{Scheme 2 (Lagrange Coded Computing in $L$ groups)}\label{sec::scheme2}

The Master Node splits the Worker Nodes into $L$ many groups, $G_1,\ldots,G_L$, where $G_j$ contains $d_j$ Worker Nodes and $\sum_{j\in[L]} d_j = d$ . Let $\alpha_{j_i} = \alpha_{i+\sum_{k<j} d_k}$. Within each group, $G_j$, the Master Node sends $f_j(\alpha_{j_i})$ to the $i$-th Worker Node in it, asks the worker node to compute 
\begin{align}\label{eq::workernode_h_scheme2}
h_j\left(\alpha_{j_i}\right) = \psi_j\left(f_j\left(\alpha_{j_i}\right)\right),
\end{align}
and sends it back to the Master Node. By doing this, the Master Node expects to obtain from $G_j$,
\begin{align}\label{resp::schemeI}
    \left(h_j\left(\alpha_{j_1}\right), \ldots, h_j\left(\alpha_{j_{d_j}}\right)\right),
\end{align}
and it is a codeword of an $\rs\left(d_j, K^{(j)}\right)$, where $K^{(j)} = D_j(\ell_j+T-1)+1$. Note that, recovering all $h_j$'s, $j\in[L]$, gives us all required computations. However, recovering $h_j$ from $G_j$ only gives us the required computations which involve $\psi_j$. This fact undermines the ability of the system to tolerate stragglers.  The values of performance metrics are formulated in the theorem below.

\begin{theorem}[Scheme 2]\label{thm:schemeI}
    Fix $A$ and $T$, and for each $j\in[L]$, set $K_2^{(j)} \triangleq D_j(\ell_j+T-1)+1$. Further, for each $j\in[L]$, choose $d_j \ge K_2^{(j)} + 2A$.

    Suppose that there are $d$ Worker Nodes, of which at most $A$ are adversarial and at most $T$ are colluding. Then Scheme 2 is $T$-secure, correct and achieves the following metrics.
    \begin{itemize}[leftmargin=*]
        \item Straggler Resistance: $\max\limits_{d_1,\ldots,d_L}\left\{\min\limits_{j\in[L]}\left\{d_j-K_2^{(j)}-2A\right\}\right\}$.
        \item Upload Cost: $\sum_{j\in[L]}d_js_j$.
        \item Download Cost: $\sum_{j\in[L]}\left(K_2^{(j)}+2A\right)$.
        \item Computation in MN: $(\ell_j+T)s_j$ multiplications of $\cdot$'s for group $G_j$.
        \item Computation in WN: no multiplications of $\cdot$'s and the number of multiplications $\times$'s in each of $\psi_j$ to compute.
    \end{itemize}
\end{theorem}

\begin{proof}
    We consider the worst-case scenario when all $A$ adversaries and $T$ colluding workers are in the same group. Since $K^{(j)}-1$ corresponds to the degree of the polynomial $h_j$ defined in \eqref{eq::workernode_h_scheme2} and Worker Node $i$ computes $h_j\left(\alpha_{j_i}\right)$, the $d_j$-tuple of computations \eqref{resp::schemeI} can be viewed as a codeword of an $\rs\left(d_j, K_2^{(j)}\right)$ code. Hence, from each group $G_j$, with $K_2^{(j)}+2A$ responses, the Master Node can correctly recover $h_j$, even in the presence of $A$ adversarial responses (see for example~\cite[Ch.~6]{assympdecode}).

    \noindent \textit{Stragglers Resistance.} In each group $G_j$, since we require $K_2^{(j)}+2A$ responses, we can tolerate $d_j-K_2^{(j)}-2A$ stragglers. But, the Master Node needs to recover from all groups. Hence, given $d_1,\ldots,d_L$, it can only tolerate $\min\limits_{j\in[L]}\left\{d_j-K_2^{(j)}-2A\right\}$. However, the Master Node has the freedom to decide how to split Worker Nodes into $L$ groups to achieve the highest resistance to stragglers. Hence, this scheme can tolerate $\max\limits_{d_1,\ldots,d_{L}}\left\{\min\limits_{j\in[L]}\left\{d_j-K_2^{(j)}-2A)\right\}\right\}$ stragglers.

    \noindent \textit{Upload and Download Cost.} For each group $G_j$, the Master Node sends $f_j(\alpha_{j_i})$ to the $i$-th Worker Node in it while each of them contains $s_j$ elements of $\mathbb{U}$. So, the Upload Cost is $\sum_{j\in[L]} d_js_j$ elements of $\mathbb{U}$. To do the recovery, the Master Node downloads $K_2^{(j)}+2A$ elements of $\mathbb{V}$ from each group $G_j$. As a result, the Download Cost is $\sum_{j\in[L]} \left(K_2^{(j)}+2A\right)$.
    
    \noindent \textit{Computation in MN.} The Master Node computes the values of sharing polynomials at some evaluation point, which interpolates $\ell_j$ points of dimension $s_j$ and $T$ random points of the same dimension. Assuming that all multiplications of field elements are pre-computed), the Master Node performs $(\ell_j+T)s_j$ scalar multiplications for each node in group $G_j$.

    \noindent \textit{Computation in WN.} The Worker Nodes in group $G_j$ apply $\psi_j$ to their received shared information. So, the number of $\times$ performed by the Worker Nodes in group $G_j$ is equal to the number of $\times$ in $\psi_j$.

    \noindent \textit{Security. } Since the security proofs of all Schemes~1--4 are similar, we only provide a detailed proof for Theorem~\ref{theorem:1} (see Section~\ref{sec:scheme3}).
\end{proof}

\subsection{Scheme 3 (Lagrange Coded Computing in L rounds)}\label{sec:scheme2}

The Master Node sends $\left(f_1(\alpha_i),\ldots,f_{L}(\alpha_i)\right)$ to the $i$-th Worker Node and requests to compute the values
\begin{align}\label{eq::scheme3_h}
    \pmb{h}(\alpha_i)&=(h_j(\alpha_i))_{j\in[L]}^T\notag\\
    &=\left(\psi_1(f_1(\alpha_i)),\psi_2(f_2(\alpha_i)),\ldots,\psi_L(f_L(\alpha_i))\right)^T.
\end{align}
As a result, the computations of all $d$ involved Worker Nodes can be represented as
{\small
\begin{align}\label{resp::scheme3}
    (\pmb{h}(\alpha_1),\ldots,\pmb{h}(\alpha_d))
    =\left(\begin{pmatrix}
    h_1(\alpha_1)\\
    h_2(\alpha_1)\\
    \vdots\\
    h_L(\alpha_1)
    \end{pmatrix},\ldots,
    \begin{pmatrix}
    h_1(\alpha_d)\\
    h_2(\alpha_d)\\
    \vdots\\
    h_L(\alpha_d)
    \end{pmatrix}\right).
    \end{align}
}%
\noindent The $j$-th row of \eqref{resp::scheme3} is a codeword of an $\rs(d, K^{(j)})$ code, where $K^{(j)} = D_j(\ell_j+T-1)+1$. 
Our aim is to recover all $h_j$'s, and hence, the required computations. 
There are two ways to perform the recovery, the Worker Node sends its computations all at once or one by one ($j$-th computation in round $j$). For both approaches, the Master Node requires the same amount of responsive Worker Nodes to recover $\pmb{h}$. However, some $h_j$ might have a lower degree than the others, hence, it requires fewer responses to recover. This leads to a higher Download Cost if we perform the first approach, in comparison to the second approach. The values of performance metrics (by using the second approach) are formulated in the theorem below.

\begin{theorem}[Scheme 3]\label{thm:scheme2} Fix $A$ and $T$, and for each $j\in[L]$, set $K_3^{(j)}\triangleq D_j(\ell_j+T-1)+1$. Furthermore, choose $d\ge \max\limits_{j\in[L]}K_3^{(j)} + 2A$.
Suppose that there are $d$ Worker Nodes, of which at most $A$ are adversarial and at most $T$ are colluding. Then Scheme~3 is $T$-secure, correct and achieves the following metrics.
\begin{itemize}[leftmargin=*]
    \item Straggler Resistance: $d-\max\limits_{j\in[L]}K_3^{(j)}-2A$.
    \item Upload Cost: $d\sum_{j\in[L]} s_j$
    \item Download Cost: $\sum_{j\in[L]} \left(K_3^{(j)}+2A\right)$
    \item Computation\,in\,MN: $\sum_{j\in[L]}(\ell_j+T)s_j$ multiplications of $\cdot$'s.
    \item Computation\,in\,WN: no multiplications of $\cdot$'s and the total number of multiplications $\times$'s in all $\psi_j$ to compute.
\end{itemize}
    
\end{theorem}
\begin{proof}
    We consider the worst-case scenario when we want to tolerate $A$ adversaries and $T$ colluding workers in each round. Since $K^{(j)}-1$ corresponds to the degree of polynomial defined in \eqref{eq::scheme3_h} and Worker Node $i$ computes $h_j(\alpha_i)$ in the $j$-th round, the $d_j$-tuple of computations in the $j$-th row of \eqref{resp::scheme3} can be viewed as a codeword of an $\rs\left(d_j,K^{(j)}_3\right)$ code. Hence, in round $j$, with $K_3^{(j)}+2A$ responses, the Master Node can correctly recover $h_j$, even in the presence of $A$ adversarial responses (see for example~\cite[Ch.~6]{assympdecode}).

    \noindent\textit{Stragglers Resistance.} The Master Node performs $L$ rounds of downloading phase. Hence, the Master Node requires $\max\limits_{j\in[L]}\left\{K_3^{(j)}+2A\right\}$ nodes to respond. In other words, the scheme tolerates $d-\max\limits_{j\in[L]}\left\{K_3^{(j)}\right\}-2A$ stragglers.

    \noindent\textit{Upload and Download Cost.} The Master Node sends $f_1\left(\alpha_{i}\right),\ldots,f_L\left(\alpha_{i}\right)$ to the $i$-th Worker Node while each $f_j\left(\alpha_{i}\right)$ contains $s_j$ elements of $\mathbb{U}$. So, the Upload Cost is $d\sum_{j\in[L]} s_j$ elements of $\mathbb{U}$. To do the recovery, in each round $j\in[L]$, the Master Node downloads $K_3^{(j)}+2A$ elements of $\mathbb{V}$. As a result, the Download Cost is $\sum_{j\in[L]} \left(K_3^{(j)}+2A\right)$.

    \noindent\textit{Computation in MN}. The Master Node computes the values of sharing polynomials at some evaluation point. To compute one value of $f_j$, which interpolates $\ell_j$ points of dimension $s_j$ and $T$ random points of the same dimension. Assuming that all multiplications of field elements are pre-computed, the Master Node performs $\sum_{j\in[L]}(\ell_j+T)s_j$ scalar multiplications for each node.

    \noindent\textit{Computation in WN}. The Worker Nodes apply $\psi_1,\ldots,\psi_L$ to their received shared information. So, the number of $\times$ performed by the Worker Nodes is equal to the number of $\times$ in all $\psi_j$.

    \noindent \textit{Security.} Since the security proofs of all Schemes~1--4 are similar, we only provide a detailed proof for Theorem~\ref{theorem:1} (see Section~\ref{sec:scheme3}).
\end{proof}

\subsection{Scheme 4 (Lagrange Coded Computing in One Round)}\label{sec:scheme3}
The Master Node sends $\left(f_1(\alpha_i),\ldots,f_{L}(\alpha_i)\right)$ to the $i$-th Worker Nodes and asks the $i$-th Worker Node to compute
\begin{align}\label{eq:scheme3h}
    h(\alpha_i) = \sum_{j\in[L]}\psi_j(f_j(\alpha_i))\prod_{\beta\in\pmb{\beta}\setminus\pmb{\beta}_j}(\alpha_i-\beta).
\end{align}
The Master Node expects to obtain the following values
\begin{align}\label{eq:codeword}
\left(h(\alpha_1),h(\alpha_2),\ldots,h(\alpha_{d})\right).
\end{align}
and it is a codeword of an $\rs(d,K)$ code, where $K = \max\limits_{j\in[L]}\left\{D_j(\ell_j+T-1)+\ell-\ell_j\right\}+1$. It can be verified that for all $i\in[\ell_j], j\in[L]$ we have
\begin{align}\label{eq:fact}
    h\left(\beta_{i}^{(D,j)}\right) = \gamma_i^{(j)}\psi_j\left(\bX_{i}^{(j)}\right),
\end{align}
for some constants $\gamma_i^{(j)}\in\FF$. Hence, by recovering $h$, the Master Node can obtain all required computations. The values of performance metrics are formulated in the theorem below.

\begin{theorem}[Scheme 4]\label{theorem:1}
    Fix $A$ and $T$, and set $K_4\triangleq \max\limits_{j\in[L]}\left\{D_j(\ell_j+T-1)+\ell-\ell_j\right\}+1$. Further, choose $d\ge K_4+2A$.
    Suppose that there are $d$ Worker Nodes, of which at most $A$ are adversarial and at most $T$ are colluding. Then Scheme~4 is $T$-secure, correct, and achieves the following metrics.
    \begin{itemize}[leftmargin=*]
        \item Straggler Resistance: $d-K_4-2A$.
    \item Upload Cost: $d\sum_{j\in[L]} s_j$.
    \item Download Cost: $K_4+2A$.
    \item Computation\,in\,MN: $\sum_{j\in[L]}(\ell_j+T)s_j$ multiplications of $\cdot$'s.
    \item Computation\,in\,WN: $L$ multiplications of $\cdot$'s and the total number of multiplications $\times$'s in all $\psi_j$ to compute.
    \end{itemize}
\end{theorem}

\begin{proof}
Since $K-1$ corresponds to the degree of the polynomial $h$ defined in \eqref{eq:scheme3h} and Worker Node $i$ computes $h(\alpha_i)$, the $d$-tuple of computations $\eqref{eq:codeword}$ can be viewed as a codeword of an $\rs(d,K_4)$ code. Hence, with $K_4+2A$ responses, the Master Node can correctly recover the polynomial $h$, even in the presence of $A$ adversarial responses (see for example~\cite[Ch.~6]{assympdecode}).

\noindent\textit{Stragglers Resistance}. Since we require $K_4+2A$ responses, we can tolerate $d-K_4-2A$ stragglers.

\noindent\textit{Upload and Download Cost}. The Master Node sends $f_1\left(\alpha_{i}\right),\ldots,f_L\left(\alpha_{i}\right)$ to the $i$-th Worker Node while each $f_j\left(\alpha_{i}\right)$ contains $s_j$ elements of $\mathbb{U}$. So, the Upload Cost is $d\sum_{j\in[L]} s_j$ elements of $\mathbb{U}$. To do the recovery, the Master Node needs to download $K_4+2A$ elements of $\mathbb{V}$.

\noindent\textit{Computation in MN}. The Master Node computes the values of sharing polynomials at some evaluation point. To compute one value of $f_j$, which interpolates $\ell_j$ points of dimension $s_j$ and $T$ random points of the same dimension. Assuming that all multiplications of field elements are pre-computed, the Master Node performs $\sum_{j\in[L]}(\ell_j+T)s_j$ scalar multiplications for each node.

\noindent\textit{Computation in WN}. The Worker Nodes apply $\psi_1,\ldots,\psi_L$ to the received shared information. So, the number of $\times$ performed by the Worker Nodes is equal to the number of $\times$ in all $\psi_j$. In addition, the Worker Nodes also need to do $L$ scalar multiplications to compute the necessary value of $h$.

\noindent\textit{Security}. The notion of security is equivalent to zero-mutual information between the dataset $\bX$ and $T$ values of sharing polynomials. We do note that, for all $j\in [L]$, the shared data can be written as
{\begin{align}\label{eq:rep}
    [f_j(\alpha_1),&\ldots,f_j(\alpha_d)]\notag\\&= \left[\bX_1^{(j)},\ldots,\bX_{\ell_j}^{(j)},\bZ_1^{(j)},\ldots,\bZ_T^{(j)}\right]\cdot G^{(j)}.
\end{align}}%
We can split $G^{(j)}$ into two parts, 
    $\begin{pmatrix}
    G^{(j,\text{top})}\\
    G^{(j,\text{bottom})}
    \end{pmatrix}$, such that, for $r\in[\ell_j],c\in[d]$,
{\begin{align}
    &G^{(j,\text{top})}_{r,c}\notag\\ &= \frac{\prod_{k\ne r} \left(\alpha_c-\beta_k^{(D,j)}\right)}{\prod_{k\ne r} \left(\beta_r^{(D,j)}-\beta_k^{(D,j)}\right)}\cdot\frac{\prod_{k\in[T]} \left(\alpha_c-\beta_k^{(R)}\right)}{\prod_{k\in[T]}\left(\beta_r^{(D,j)}-\beta_k^{(R)}\right)},
    \end{align}}
and for $r \in[T],c\in[d]$,
\begin{align}
    &G^{(j,\text{bottom})}_{r,c}\notag\\ &= \frac{\prod_{k\in[\ell_j]}\left(\alpha_c-\beta_k^{(D,j)}\right)}{\prod_{k\in[\ell_j]}\left(\beta_r^{(R)}-\beta_k^{(D,j)}\right)}\cdot\frac{\prod_{k\ne r}\left(\alpha_c-\beta_k^{(R)}\right)}{\prod_{k\ne r}\left(\beta_r^{(R)}-\beta_k^{(D,j)}\right)},
\end{align}
where $G^{(j,\cdot)}_{r,c}$ is the element in the $r$-th row and $c$-th column of the matrix $G^{(j,\cdot)}$.
Then, we can rewrite \eqref{eq:rep} as
\begin{align}
    &[f_j(\alpha_1),\ldots,f_j(\alpha_d)]\notag  \\ &=\left[\bX_1^{(j)},\ldots,\bX_{\ell_j}^{(j)}\right]\cdot G^{(j,\text{top})}+ \left[\bZ_1^{(j)},\ldots,\bZ_T^{(j)}\right]\cdot G^{(j,\text{bottom})}.
\end{align}
Let $\Vec{\bX}^{(j)} = \left[\bX_1^{(j)},\ldots,\bX_{\ell_j}^{(j)}\right]$, $\Vec{\bZ}^{(j)} =\left[\bZ_1^{(j)},\ldots,\bZ_T^{(j)}\right]$, and $\pmb{Y}^{(j)}_T$ be any $T$ components of $[f_j(\alpha_1),\ldots,f_j(\alpha_d)]$ available to colluding nodes, for simplicity of notation, we define $\pmb{Y}^{(j)}_T$ as the first $T$ components of the shared data. By applying the chain rule for the mutual information and by noting that the entropy of any file is equal to the number of elements in it by assuming source coding is applied before, we have
\begin{alignat}{2}\label{eq:sec}
    &I\left(\Vec{\bX}^{(j)};\pmb{Y}^{(j)}_T\right)\notag\\ 
    &= I\left(\vX,\vZ;\pmb{Y}_T\right) - I\left(\vZ;\pmb{Y}_T|\vX\right)\notag\\
    &= H\left(\vX,\vZ\right) - H\left(\vX,\vZ|\pmb{Y}_T\right)\notag\\&\quad\quad- H\left(\vZ|\vX\right) + H\left(\vZ|\vX,\pmb{Y}_T\right)\notag\\
    &= H\left(\vX\right)+H\left(\vZ\right) - H\left(\vX,\vZ|\pmb{Y}_T\right) \notag\\&\quad\quad- H\left(\vZ|\vX\right) + H\left(\vZ|\vX,\pmb{Y}_T\right)\notag\\
    &= \ell_j + T -H\left(\vX,\vZ|\pmb{Y}_T\right)\notag\\ &\quad\quad- H\left(\vZ|\vX\right) + H\left(\vZ|\vX,\pmb{Y}_T\right)\notag\\
    &= \ell_j + T -H\left(\vX,\vZ|\pmb{Y}_T\right)\notag\\ &\quad\quad- H\left(\vZ\right) + H\left(\vZ|\vX,\pmb{Y}_T\right)\notag\\
    &=\ell_j -H\left(\vX,\vZ|\pmb{Y}^{(j)}_T\right) + H\left(\vZ|\vX,\pmb{Y}^{(j)}_T\right).
\end{alignat}
Note that, 
\begin{align}
    &\pmb{Y}^{(j)}_T = \left[\vX,\vZ\right]\cdot G^{(j)}_T\notag\\&\implies H\left(\vX,\vZ|\pmb{Y}^{(j)}_T\right) = \ell_j+T-\rank\left(G^{(j)}_T\right),
\end{align}
where $G^{(j)}_T$ is the first $T$ columns of $G$. This is because, given $Y_T^{(j)}$, all solutions $\left(\vX,\vZ\right)$ are equaly likely. We also have
\begin{align}
    &\vZ \cdot G_T^{(j,\text{bottom})} = \pmb{Y}^{(j)}_T - \vX\cdot G_T^{(j,\text{top})}\\&\implies H\left(\vZ|\vX,\pmb{Y}^{(j)}_T\right) = T-\rank\left(G_T^{(j,\text{bottom})}\right).
\end{align}
This is because, given $\left(\vX,\pmb{Y}^{(j)}_T\right)$, all solutions $\vZ$ are equally likely. Therefore, \eqref{eq:sec} can be expressed as
\begin{align}
I&\left(\Vec{\bX}^{(j)};\pmb{Y}^{(j)}_T\right)\notag \\
    &= \ell_j -  \left(\ell_j+T-\rank\left(G^{(j)}_T\right)\right) + \left( T-\rank\left(G_T^{(j,\text{bottom})}\right)\right)\notag\\
    &= \rank\left(G^{(j)}_T\right) - \rank\left(G_T^{(j,\text{bottom})}\right)\notag\\
    &\le T - \rank\left(G_T^{(j,\text{bottom})}\right)\notag\\
    &=0,
\end{align}
since $G_T^{(j,\text{bottom})}$ is invertible (see \cite{assympdecode}). 
\end{proof}

\begin{example}[Computation Scheme 4]\label{example::scheme3}
    Consider the task where the Master Node has the dataset $\bX = (A, B, C)$, where $A,B,C\in \FF^{n\times n}$ for some finite field $\FF$ and wants to obtain
    $$
    A^2,\quad B^2,\quad AB + BC^2
    $$
    by employing $d$ Worker Nodes ensuring security against one colluding node and tolerating one adversarial node. The Master Node wants to find $\ell = 3$ values of evaluations for $L = 2$ functions $\psi_1(A) = A^2$, $\psi_1(B) = B^2$, and $\psi_2(A,B,C) = AB + BC^2$. By considering a set of distinct evaluation points $\pmb{\beta} = \left\{\beta_1^{(D,1)},\beta_2^{(D,1)},\beta_1^{(D,2)},\beta_1^{(R)}\right\}$, the Master Node constructs $L = 2$ sharing polynomials satisfying
    \begin{align}
    \begin{cases}
        f_1\left(\beta_1^{(D,1)}\right) = A,\, f_1\left(\beta_2^{(D,1)}\right) = B,\, f_1\left(\beta_1^{(R)}\right) = Z_1,\\
        f_2\left(\beta_1^{(D,2)}\right) = (A,B,C),\, f_2\left(\beta_1^{(R)}\right) = (Z_2,Z_3,Z_4),\\
    \end{cases}
    \end{align}
    where $Z_1,\ldots,Z_4\in\FF^{n\times n}$ are independent random elements uniformly distributed over the same alphabets as $X_i$'s. The polynomials $f_1$ and $f_2$ can be obtained by Lagrange Interpolation. We can see that $f_1$ is a polynomial of degree $2$ and $f_2$ is a polynomial of degree $1$. The Master Node also considers another set of evaluation points $\pmb{\alpha}=\{\alpha_1,\ldots,\alpha_d\}$, $\pmb{\alpha}\cap\pmb{\beta} = \emptyset$ and assigns each Worker Node with a unique evaluation point from $\pmb{\alpha}$. The Master Node sends $(f_1(\alpha_i),f_2(\alpha_i))$ to the $i$-th Worker Node and asks it to compute
    \begin{align}\label{eq::h_in_example1}
        h(\alpha_i) = &\psi_1(f_1(\alpha_i))\left(\alpha_i - \beta_1^{(D,2)}\right)\notag \\&+\psi_2(f_2(\alpha_i)) \left(\alpha_i - \beta_1^{(D,1)}\right) \left(\alpha_i - \beta_2^{(D,1)}\right).
    \end{align}
    As a result, the Master Node expects to obtain
    \begin{align}
        (h(\alpha_1),\ldots, h(\alpha_d)),
    \end{align}
    which is a codeword of an $\rs(d,6)$ code. Hence, with any $8$ responses, the Master Node can recover $h$, in the presence of one adversarial response. We can easily check that
    \begin{align}
        \begin{cases}
            h\left(\beta_1^{(D,1)}\right) &= \gamma_{1}^{(1)} A^2,\vspace{2mm}\\
            h\left(\beta_2^{(D,1)}\right) &=\gamma_{2}^{(1)} B^2,\vspace{2mm}\\
             h\left(\beta_1^{(D,2)}\right) &=\gamma_{1}^{(2)} (AB+ BC^2).
        \end{cases}
    \end{align}
    where $\gamma_{i}^{(j)}$ is some field element $\FF$ (which can be pre-computed).
\end{example}

\section{Results Verification}\label{sec::verification}

In this section, we study schemes that verify the correctness of workers' computations without asking extra nodes to respond. 
In particular, we propose two verification schemes. 
In the first verification scheme, the Master Node randomly generates a nonzero field element $\textrm{v}$ as a private verification key, 
and constructs an additional sharing polynomial by incorporating $\textrm{v}$.
Hence, from the workers' responses, the Master Node computes two values and then uses the private key $\textrm{v}$ to certify correctness.
We remark that similar methods were employed in~\cite{colombo2023,kruglik2023} to enable results verification in private information retrieval protocols.
This verification scheme works for general computations. 

In contrast, our second verification scheme is restricted to computations on square matrices only.
This scheme adapts the famous \textit{Freivalds}' algorithm \cite{freivalds1979} -- a probabilistic randomized algorithm used to verify matrix multiplication.
We remark that there is prior work that adapts Freivalds' algorithm for verification purposes in distributed computing schemes~\cite{rawad2022, tang2022}. 
However, in these works~\cite{rawad2022, tang2022}, to verify the response of a Worker Node, the Master Node performs Frievalds' algorithm for each worker.
On the other hand, our approach outsources these computations to the servers and significantly reduces the computation load for the Master Node. Moreover, as we argue in Section~\ref{sec::verification-B}, the increase in workload for the Worker Nodes is negligible. 
Unfortunately, the drawback of the scheme is that we are unable to identify the malicious nodes (unlike those in ~\cite{rawad2022, tang2022}).

For expository purposes, we discuss our verification methods with respect to Scheme 4 in Section~\ref{sec:scheme3}. 
Nevertheless, these verification techniques are also applicable to any other schemes in Section~\ref{sec::coded_computation}. Before we proceed to the general case, we consider the scenario in Example~\ref{example::scheme3} and highlight the main ideas of our approaches.

\begin{example}[Verification Scheme 4A]
Consider the same setup as in Example~\ref{example::scheme3}. 
To perform the verification, the Master Node generates a random non-zero field element $\textrm{v}$ uniformly distributed over $\FF$ and creates the following additional tasks.
\begin{align*}
\begin{cases}
\psi_1^{(\mathrm{v})}(\mathrm{v}A) &= \mathrm{v}^2A^2,\\ \psi_1^{(\mathrm{v})}(\mathrm{v}B) &= \mathrm{v}^2B^2,\\
\psi_2^{(\mathrm{v})}(\mathrm{v}^3A,\mathrm{v}^3B, \mathrm{v}^2B, \mathrm{v}^2C) &= \mathrm{v}^6 (AB+BC^2),
\end{cases}
\end{align*}

Specifically, these tasks involve the following  polynomials
\begin{align}
            \psi_1^{(\mathrm{v})}(x) &= x^2,\\
            \psi_2^{(\mathrm{v})}(x,y,z_1,z_2) &= xy + z_1z_2^2.
    \end{align}

In other words, the Master Node gives an additional task to the Worker Nodes, with new shares $(f_1^{(\mathrm{v})}(\alpha_i),f_2^{(\mathrm{v})}(\alpha_i))$, to compute new polynomial $h^{(\mathrm{v})}(\alpha_i)$ which is similar to~\eqref{eq::h_in_example1}, but with $\psi_1^{(\mathrm{v})}(f_1^{(\mathrm{v})}(\alpha_i))$ and $\psi_2^{(\mathrm{v})}(f_2^{(\mathrm{v})}(\alpha_i))$. 
As the degrees of $h$ and $h^{(\mathrm{v})}$ are the same, the Master Node does not require more responding nodes to obtain these extra computations. 
This allows the Master Node to perform verification of the required computations by checking that the following holds:
\begin{align}
    \begin{cases}
        \psi_1^{(\mathrm{v})}(\mathrm{v}A) &= \mathrm{v}^2 \psi_1(A)\\
        \psi_1^{(\mathrm{v})}(\mathrm{v}B) &= \mathrm{v}^2\psi_1(B)\\
        \psi_2^{(\mathrm{v})}(\mathrm{v}^3A,\mathrm{v}^3B, \mathrm{v}^2B, \mathrm{v}^2C) &= \mathrm{v}^6 \psi_2(A,B,C).
    \end{cases}
\end{align}

Note that, $\psi_2^{(\mathrm{v})}$ has $4$ inputs, while initial $\psi_2$ (in Example~\ref{example::scheme3}) only has $3$ inputs. This means that to perform verification with this technique, the Upload Cost increases by a factor slightly greater than two. This phenomenon also occurs in general. 
This approach also doubles the workload of the Worker Nodes and the download cost.
To lower the workload and download cost, we propose a second approach.

    \end{example}
\begin{example}[Verification Scheme 4B]
    Consider the same setup as in Example~\ref{example::scheme3}. 
    To perform the required computations and verify their correctness, the Master Node modifies the input matrices. 
    The exact form of modification depends on the structure of the computations. 
    Let $\pmb{0}\triangleq (0,0,\ldots,0)^T$. The Master Node generates random vectors $\textrm{u}_1,\textrm{v}_1,\textrm{u}_2,\textrm{v}_2,\textrm{u}_3,\textrm{v}_3\in \FF^{n}$ uniformly distributed over $\FF^n\setminus\{\pmb{0}\}$. We construct
    \begin{enumerate}[label=(\arabic*)]
        \item $A_1^{(1)} = \begin{bmatrix}
        A & \pmb{0}\\
        \textrm{u}_1^T A & 0
    \end{bmatrix}\text{ and } A_2^{(1)} = \begin{bmatrix}
        A & A\textrm{v}_1\\
        \pmb{0}^T & 0
    \end{bmatrix}
    $ for the computation $A^2$,\vspace{2mm}
    \item $B_1^{(2)} = \begin{bmatrix}
        B & \pmb{0}\\
        \textrm{u}_2^T B & 0
    \end{bmatrix}\text{ and } B_2^{(2)} = \begin{bmatrix}
        B & B\textrm{v}_2\\
        \pmb{0}^T & 0
    \end{bmatrix}
    $ for the computation $B^2$,\vspace{2mm}
    \item $A_1^{(3)} =\begin{bmatrix}
        A & \pmb{0}\\
        \textrm{u}_3^TA & 0
    \end{bmatrix}$, $B_2^{(3)} = \begin{bmatrix}
        B & B\textrm{v}_3\\
        \pmb{0}^T & 0
    \end{bmatrix}$, $B_1^{(3)} = \begin{bmatrix}
        B & \pmb{0}\\
        \textrm{u}_3^T B & 0
    \end{bmatrix}$, $C_2^{(3)} = \begin{bmatrix}
        C & \pmb{0}\\
        \pmb{0}^T & 0
    \end{bmatrix}$, and $C_3^{(3)} = \begin{bmatrix}
        C & C\textrm{v}_3\\
        \pmb{0}^T & 0
    \end{bmatrix}$ for the computation $AB + BC^2$.
    \end{enumerate}
    We can easily see that
    \begin{align*}
        A_1^{(1)}A_2^{(1)} &= \begin{bmatrix}
            A^2 & A^2\textrm{v}_1\\
            \textrm{u}_1^T A^2 & \textrm{u}_1^TA^2\textrm{v}_1
        \end{bmatrix}\vspace{3mm} \\
        B_1^{(2)}B_2^{(2)} &= \begin{bmatrix}
            B^2 & B^2\textrm{v}_2\\
            \textrm{u}_2^T B^2 & \textrm{u}_2^TB^2\textrm{v}_2
        \end{bmatrix}
        \end{align*}
        and
        \begin{align*}
        A_1^{(3)}&B_2^{(3)} + B_1^{(3)}C_2^{(3)}C_3^{(3)} \\&= \begin{bmatrix}
            AB & AB\textrm{v}_3\\
            \textrm{u}_3^TAB & \textrm{u}_3^TAB\textrm{v}_3
        \end{bmatrix} + \begin{bmatrix}
            BC^2 & BC^2\textrm{v}_3\\
            \textrm{u}_3^TBC^2 & \textrm{v}_3^TBC^2\textrm{v}_3
        \end{bmatrix} \\
        &= \begin{bmatrix}
            AB+BC^2 & (AB+BC^2)\textrm{v}_3\\
            \textrm{v}_3^T(AB+BC^2) & \textrm{u}_3^T(AB+BC^2)\textrm{v}_3
        \end{bmatrix}.
    \end{align*}
    
    All the above computations contain $A^2, B^2$ and $AB+ BC^2$ as required.
    The remaining components are functions of these computation results and so,
    we use $\textrm{u}_1,\textrm{v}_1,\textrm{u}_2,\textrm{v}_2,\textrm{u}_3,\textrm{v}_3$ used as verification keys. 
    Hence, we use the remaining components to perform verification. 
    To implement this framework, we have to construct our sharing polynomials based on the modified matrices instead. 
    Note that, the dimension of input matrices does not affect the number of required responses. 
    However, this approach changes the sharing polynomials. 
    In this example, for the initial computations, we have $\psi_1(x) = x^2$ and $\psi_2(x,y,z) = xy+yz^2$. 
    But with modified inputs, we have $\hat{\psi}_1(x,y) = xy$ and $\hat{\psi}_2(x,y,a,b,c) = xy+abc$. Hence, the Master Node can apply Scheme~4, to obtain the computations
    \begin{align}
            &\hat{\psi}_1\left(A_1^{(1)},A_2^{(1)}\right) = A_1^{(1)}A_2^{(1)}\\
            &\hat{\psi}_1\left(B_1^{(1)},B_2^{(1)}\right) = B_1^{(1)}B_2^{(1)}\\
            &\hat{\psi}_2\left(A_1^{(3)},B_2^{(3)}, B_1^{(3)},C_2^{(3)},C_3^{(3)}\right)\notag\\
            &\hspace{24mm}= A_1^{(3)}B_2^{(3)} + B_1^{(3)}C_2^{(3)}C_3^{(3)}
    \end{align}

    Suppose that we want to verify the correctness of $A^2$. 
    The Master Node considers the $n\times n$ matrix on the top left corner of the recovered result $A_1^{(1)}A_2^{(1)}$.
    Let us call this matrix as $M$ and then check all of the following.
    \begin{enumerate}[label=(\arabic*)]
        \item Compute $M\textrm{v}_1$ and check if $M\textrm{v}_1$ is equal to the first $n$ elements of the last column,
        \item Compute $\textrm{u}_1^T M$ and check if $\textrm{u}_1^T M$ is equal to the first $n$ elements of the last row, and
        \item Compute $\textrm{u}_1^T M \textrm{v}_1$ and check if $\textrm{u}_1^T M \textrm{v}_1$ is equal to the element at the bottom left.
    \end{enumerate}
    The remaining computations can also be checked in the same way. 
    Note that, to check the correctness of one computation, the Master Node only needs to perform two matrix-vector multiplications and one vector-vector multiplication. The number of such multiplications is independent of the number of worker responses.
    In contrast, the prior schemes in~\cite{rawad2022,tang2022} require the Master Node to perform three matrix-vector multiplication for {\em each} worker response.
\end{example}

\subsection{Verification Scheme 4A}
In this verification scheme, our goal is to obtain another set of computation results that are related to our initial required computations. 
This relationship is controlled by the Master Node. To do so, the Master Node generates a random nonzero element $\mathrm{v}$ uniformly distributed over $\FF$. 
In addition to the initial set of $L$ sharing polynomials~\eqref{eq:sharingpoly}, the Master Node also considers another $L$ sharing polynomials over the same set of evaluation points $\pmb{\beta}$ and $\pmb{\alpha}$.
The polynomials are chosen such that, for all $j\in[L]$, we have $f_j^{(\mathrm{v})}:\FF \to (\FF^{n\times n})^{
\hat{s}_j}$ for some $\hat{s}_j$ and 
{\small 
\begin{alignat}{4}
    &f_j^{(\textrm{v})}\left(\beta_1^{(D,j)}\right) &&= v_j\left(\bX_1^{(j)}\right)&&,\ldots, f_j^{(\textrm{v})}\left(\beta_{\ell_j}^{(D,j)}\right) &&= v_j\left(\bX_{\ell_j}^{(j)}\right),\notag\\
&f_j^{(\textrm{v})}\left(\beta_1^{(R)}\right) &&= \bZ_1^{(\textrm{v},j)}&&,\ldots,f_j^{(\textrm{v})}\left(\beta_T^{(R)}\right) &&= \bZ_T^{(\textrm{v},j)},
\end{alignat}
}%
\noindent for some functions $v_1,\ldots,v_L$. The Master Node additionally sends $(f_1^{(\textrm{v})}(\alpha_i),\ldots,f_L^{(\textrm{v})}(\alpha_i))$ to the $i$-th Worker Node and asks the $i$-th Worker Node to compute
\begin{align}\label{eq:verifh}
    h^{(\textrm{v})}(\alpha_i) = \sum_{j\in[L]} {\psi}^{(\textrm{v})}_j\left(f^{(\textrm{v})}_j(\alpha_i)\right)\prod_{\beta\in\pmb{\beta}\setminus\pmb{\beta}_j}(\alpha_i-\beta)
\end{align}
\noindent for some $\psi_j^{(\textrm{v})}$. The functions $v_1,\ldots,v_L$ and $\psi_1^{(\textrm{v})},\ldots, \psi_L^{(\textrm{v})}$ are chosen by the Master Node, so that $h^{(\textrm{v})}(\beta) = \textrm{v}^n h(\beta)$, for some $n$, for all $\beta\in\pmb{\beta}\setminus\{\beta_i^{(R)}:i\in[T]\}$. 

One possible construction of $v_1,\ldots,v_L$ is as follows. We can write each $\psi_j$ as the sum of polynomials $p_{j,i}$ of distinct degrees. Suppose that there are $r_j$ many distinct degrees polynomials on $\psi_j$, then we have
\begin{align}
    \psi_j(M_1,\ldots,M_{s_j}) & =  p_{j,1}\left(M_{(j,1)}^{(1)},\ldots, M_{(j,1)}^{(s_{j_1})}\right) + \cdots\notag\\
    & + p_{j,r_j}\left(M_{(j,r_j)}^{(1)},\ldots, M_{(j,r_j)}^{(s_{j_{r_j}})}\right).
\end{align}

For each involved polynomial $p_{j,k}$, we modify its inputs by applying a function $v_{j,k}$ so that the value of polynomial $p_{j,k}$ with the modified inputs is equal to $\mathrm{v}^{n_j}p_{j,k}\left(M_{(j,k)}^{(1)},\ldots, M_{(j,k)}^{(s_{j_k})}\right)$. Let $\delta_{j,k}$ be the degree of $p_{j,k}$. Define $\Delta_j \triangleq \lcm(\delta_{j,1},\ldots,\delta_{j,r_j})$ and set
{\small \begin{align}
    v_{j,k} \left(M_{(j,k)}^{(1)},\ldots, M_{(j,k)}^{(s_{j_k})}\right)= \mathrm{v}^{\frac{\Delta_j}{\delta_{j,k}}} \left(M_{(j,k)}^{(1)},\ldots, M_{(j,k)}^{(s_{j_k})}\right),
\end{align}}
and
\begin{align}
    v_j(M_1,\ldots,M_{s_j}) = \Bigg(&v_{j,1}\left(M_{(j,1)}^{(1)},\ldots, M_{(j,1)}^{(s_{j_1})}\right),\ldots,\notag\\&v_{j,r_j}\left(M_{(j,r_j)}^{(1)},\ldots, M_{(j,r_j)}^{(s_{j_{r_j}})}\right)\Bigg).
\end{align}

Let $\psi_j^{(\mathrm{v})}$ be the polynomial with the same addition and multiplication structure as $\psi_j$ but with $\hat{s}_j = \sum_{k=1}^{r_j}s_{j_k}$ inputs. Then, we can see that
\begin{align}
    \psi_j^{(\mathrm{v})}&\left(v_j(M_1,\ldots,M_{s_j})\right)\notag\\& = \sum_{k=1}^{r_j}p_{j,k}\left(\mathrm{v}^{\frac{\Delta_j}{\delta_{j,k}}} \left(M_{(j,k)}^{(1)},\ldots, M_{(j,k)}^{(s_{j_k})}\right)\right)\notag\\
    &= \mathrm{v}^{\Delta_j} \sum_{k=1}^{r_j} p_{j,k} \left(M_{(j,k)}^{(1)},\ldots, M_{(j,k)}^{(s_{j_k})}\right)\notag\\
    &= \mathrm{v}^{\Delta_j} \psi_j(M_1,\ldots,M_{s_j}).
\end{align}

This implies that for all $i\in[\ell_j]$, $j\in[L]$
\begin{align}
h^{(\mathrm{v})}\left(\beta_i^{(D,j)}\right) &= \gamma_i^{(j)} \psi_j^{(\mathrm{v})}\left(v_j\left(\pmb{X}_i^{(j)}\right)\right)\notag\\ &= \mathrm{v}^{\Delta_j}\gamma_i^{(j)}\psi_j\left(\pmb{X}_i^{(j)}\right)\notag\\&=  \mathrm{v}^{\Delta_j}h\left(\beta_i^{(D,j)}\right),
\end{align}
for some fixed constant $\gamma_i^{(j)}\in\FF$. We are left to show that this verification scheme can detect the presence of incorrect computations with high probability.
\begin{theorem}
    Scheme 4A can detect the presence of incorrect computations in the presence of up to $T$ adversarial nodes with probability $1-O\left(\frac{1}{q}\right)$.
\end{theorem}
\begin{proof}
Without loss of generality, let us assume that the first $T$ servers are adversarial and provide the following responses $$\left(\hat{h}(\alpha_1),\hat{h}^{(\mathrm{v})}(\alpha_1)\right),\ldots,\left(\hat{h}(\alpha_T),\hat{h}^{(\mathrm{v})}(\alpha_T)\right).$$  Let $K$ be the recovery threshold. Note that, we can write the required computations as a linear combination of correct responses, that is,
\begin{align}
        h\left(\beta_i^{(D,j)}\right) &= \sum_{k=1}^K c_{i,j,k} h(\alpha_k),\\
        h^{(\mathrm{v})}\left(\beta_i^{(D,j)}\right) &= \sum_{k=1}^K c_{i,j,k} h^{(\mathrm{v})}(\alpha_k).
\end{align}

We do note that the constants $c_{i,j,k}$ depend only on the evaluation points of participating nodes.
Specifically, they do not depend on the responses and values of $\mathrm{v}$. This comes from the fact that we can recover the polynomial $h$ by multiplying the inverse of an $K$ by $K$ Vandermonde matrix and the vector of $K$ responses. Let $\tilde{h}$ and $\tilde{h}^{(\mathrm{v})}$ be required computations affected by wrong responses from malicious nodes.
That is, 
\begin{align}
        &\tilde{h}\left(\beta_i^{(D,j)}\right)\notag\\
        &\hspace{4mm}= \sum_{k=1}^T c_{i,j,k} \hat{h}(\alpha_k) + \sum_{k=T+1}^K c_{i,j,k} h(\alpha_{k}),\\
        &\tilde{h}^{(\mathrm{v})}\left(\beta_i^{(D,j)}\right)\notag\\  &\hspace{4mm}= \sum_{k=1}^T c_{i,j,k} \hat{h}^{(\mathrm{v})}(\alpha_k) + \sum_{k=T+1}^K c_{i,j,k} h^{(\mathrm{v})}(\alpha_{k}).
\end{align}
Let $E$ be the event where the adversaries successfully persuade the Master Node to accept the wrong results.
In other words, 
\begin{align}
    E = \left(\bigcap_{i\in[\ell_j],j\in[L]} A_{i,j}\right) \cap \left(\bigcup_{i\in[\ell_j],j\in[L]} B_{i,j}\right),
\end{align}
\noindent where $A_{i,j}$ is the event when $\tilde{h}^{(\mathrm{v})} \left(\beta_i^{(D,j)}\right) = \mathrm{v}^{\Delta_j}\tilde{h}\left(\beta_i^{(D,j)}\right)$ and $B_{i,j}$ is the event when $h\left(\beta_i^{(D,j)}\right) \ne \tilde{h}\left(\beta_i^{(D,j)}\right)$. Then,
\begin{align}
    \mathbb{P}(E) &= \mathbb{P}\left(\left(\bigcap_{i\in[\ell_j],j\in[L]} A_{i,j}\right) \cap \left(\bigcup_{i\in[\ell_j],j\in[L]} B_{i,j}\right)\right)\notag \\
    &= \mathbb{P}\left(\bigcup_{i\in[\ell_j],j\in[L]}\left(\left(\bigcap_{i\in[\ell_j],j\in[L]} A_{i,j}\right)\cap B_{i,j}\right)\right)\notag\\
    &\le \sum_{i\in[\ell_j],j\in[L]} \mathbb{P}\left(\left(\bigcap_{i\in[\ell_j],j\in[L]} A_{i,j}\right)\cap B_{i,j}\right)\notag\\
    &\le \sum_{i\in[\ell_j],j\in[L]} \mathbb{P}\left(A_{i,j}\cap B_{i,j}\right)
\end{align}

Note that $A_{i,j} \cap B_{i,j}$ is equivalent to
\begin{align}\label{proof:verif_approach_1}
    0 = &\,\,\tilde{h}^{(\mathrm{v})} \left(\beta_i^{(D,j)}\right) - \mathrm{v}^{\Delta_j}\tilde{h}\left(\beta_i^{(D,j)}\right)\notag\\
    = &\,\,\tilde{h}^{(\mathrm{v})} \left(\beta_i^{(D,j)}\right) - h^{(\mathrm{v})}\left(\beta_i^{(D,j)}\right)\notag\\ &- \mathrm{v}^{\Delta_j}\left(\tilde{h}\left(\beta_i^{(D,j)}\right) - h\left(\beta_i^{(D,j)}\right)\right),
\end{align}
\noindent while $\tilde{h}\left(\beta_i^{(D,j)}\right) - h\left(\beta_i^{(D,j)}\right)$ is nonzero. The terms $\tilde{h}^{(\mathrm{v})} \left(\beta_i^{(D,j)}\right) - h^{(\mathrm{v})}\left(\beta_i^{(D,j)}\right)$ and $\tilde{h}\left(\beta_i^{(D,j)}\right) - h\left(\beta_i^{(D,j)}\right) \ne 0$ are fully controlled by adversary nodes (hence, these terms are deterministic and independent from $\mathrm{v}$). However, at most $T$ adversaries are not able to obtain any information of $\mathrm{v}$.  
Therefore, \eqref{proof:verif_approach_1} is equivalent to $G_{i,j}(\mathrm{v}) = 0$, where $G_{i,j}(\mathrm{v}) = a_{i,j} - \mathrm{v}^{\Delta_j} b_{i,j}$, $b_{i,j} \ne 0$. Hence, applying Schwartz-Zippel Lemma~\cite{SZ1,SZ2}, we have,
\begin{align}
    \mathbb{P}(E) &\le \sum_{i\in[\ell_j],j\in[L]} \mathbb{P}\left(A_{i,j}\cap B_{i,j}\right)\notag\\
    &= \sum_{i\in[\ell_j],j\in[L]}\mathbb{P}\left(G_{i,j}(\mathrm{v}) = 0\right)\notag\\
    &{\le} \sum_{i\in[\ell_j],j\in[L]} \frac{\Delta_j}{q-1}\leq O\left(\frac{1}{q}\right).
\end{align}
\end{proof}

\subsection{Verification Scheme 4B (for square matrices)}
\label{sec::verification-B}
In this verification scheme, we focus on operations over $n$ by $n$ square matrices. 
This approach is inspired by Freivalds' algorithm.
Specifically, if we want to verify whether $C = AB$, Frievalds' algorithm first generates a nonzero random vector $\textrm{v}\in\FF^n$ and checks whether $C\textrm{v} = AB\textrm{v}$. 
The key observation is that the latter verification requires only three matrix-vector multiplications with no matrix-matrix multiplications. 

Suppose that $C=AB$ is one of our required computations. 
Our objective is to outsource some extra computations to the Worker Nodes by adding one column to $B$, 
such that the modified computation contains both $C$ and $C\textrm{v}$. 
More precisely, instead of computing $AB$, we compute $A[B\,\,B\textrm{v}] = [AB\,\,AB\textrm{v}] = [C\,\,C\textrm{v}]$. 
However, this contradicts with our assumption that all inputs of the computations have the same dimension. 
Therefore, we add one extra row and one extra column on both $A$ and $B$, generate another nonzero random vector $\textrm{u}\in\FF^n$ 
and compute
{\small 
\begin{equation*}
\begin{bmatrix}
    A & \pmb{0}\\
    \textrm{u}^TA & 0
\end{bmatrix}\begin{bmatrix}
    B & B\textrm{v}\\
    \pmb{0}^T & 0
\end{bmatrix} = \begin{bmatrix}
    AB & AB\textrm{v}\\
    \textrm{u}^TAB & \textrm{u}^TAB\textrm{v}
\end{bmatrix} = \begin{bmatrix}
    C & C\textrm{v}\\
    \textrm{u}^TC & \textrm{u}^TC\textrm{v}
\end{bmatrix}.
\end{equation*}
}


By doing so, we also obtain more verification equations to check the correctness of $C$.

For the general computations on square matrices, we generalize the logic above as follows. The Master Node first modifies inputs depending on their position in the monomials of $\psi_j$. For instance, let $m_{j,i}$ be the monomials of $\psi_j$, $i\in[k_j]$. Then, we can rewrite it as
\begin{align}
&\psi_j(M_1,\ldots,M_{s_j}) = \sum_{i=1}^{k_j} m_{j,i}\left(M_{(j,i)}^{(1)},\ldots,M_{(j,i)}^{(s_{j_i})}\right),
\end{align}
\noindent where $s_{j_i}$ denotes the number of matrices involved in the monomial $m_{j,i}$. To perform such a modification, we write matrix powers as multiplications of a matrix by itself several times. For instance, we write $m(M_1,M_2) = M_1^2M_2^3$ as $M_1M_1M_2M_2M_2$. Let $M$ be a matrix in one of the monomial $m_{j,i}$. After it, the Master Node generates two nonzero random vectors $\textrm{u}_j,\textrm{v}_j\in\FF^n$ and performs modification on $M$ depending on its position in $m_{j,i}$.
\begin{enumerate}
    \item If $M$ is the first matrix in the monomial, then we modify $M$ into $M^{(\textrm{first})}=\begin{bmatrix}
        M & \pmb{0}\\
        \textrm{u}_j^TM & 0
    \end{bmatrix}$.\vspace{3mm}
    \item If $M$ is the last matrix in the monomial, then we modify $M$ into $M^{(\textrm{last})}=\begin{bmatrix}
        M & M\textrm{v}_j\\
        \pmb{0}^T & 0
    \end{bmatrix}$.\vspace{3mm}
    \item If $M$ is the first and the last in the monomial, then we modify $M$ into $M^{(\textrm{first,last})}=\begin{bmatrix}
        M & M\textrm{v}_j\\
        \textrm{u}_j^TM & \textrm{u}_j^T M \textrm{v}_j
    \end{bmatrix}$. This case only happens when $m_{j,i}(x) = x$. However, if the monomial $m_{j,i}(x) = x$, the Master Node does not need to include this monomial in the distributed computation scheme as it can be done easily by the Master Node. Hence, we ignore this case.\vspace{3mm}
    \item If $M$ is neither the first nor the last matrix in the monomial, then we modify $M$ into $M^{(\textrm{middle})}=\begin{bmatrix}
        M & \pmb{0}\\
        \pmb{0}^T & 0
    \end{bmatrix}$.
\end{enumerate}

We do note that we also need to slightly modify the monomials as the number of inputs might change. 
Let us look at the same example above when $m(M_1,M_2) = M_1M_1M_2M_2M_2$. After modification our goal is to compute $$M_1^{(\text{first})}M_1^{(\text{mid})}M_2^{(\text{mid})}M_2^{(\text{mid})}M_2^{(\text{last})}.$$

There are only two inputs in the initial monomial $m$.
However, in the modified monomial, we have four inputs $M_1^{(\text{first})},M_1^{(\text{mid})},M_2^{(\text{mid})},$ and $M_2^{(\text{last})}$. 
We denote the modified polynomial as $\hat{m}$. 
In Propositions~\ref{thm:mono_with_mod_inp} and~\ref{thm:polywithmodif},
we show that such modifications on inputs and monomials result in a matrix that contains the required computations. 
After that, we prove verification guarantees in Theorem~\ref{thm::verificationB}.

\begin{proposition}\label{thm:mono_with_mod_inp}
    Let $m$ be a monomial in which inputs are matrices $M_1,\ldots,M_n$. Let us write $m$ as multiplications of matrices (of power one). Let $\hat{M}_1,\ldots,\hat{M}_{\hat{n}}$ be the modified matrices with nonzero random vectors $\mathrm{u},\mathrm{v}\in\FF^n$ and let $\hat{m}$ be the monomial with the same multiplication structure as $m$, but with $\hat{n}$ inputs. Then,
    \begin{align}
        &\hat{m}(\hat{M}_1,\ldots,\hat{M}_{\hat{n}})\notag\\ &= \begin{bmatrix}
            m(M_1,\ldots,M_n) & m(M_1,\ldots,M_n)\mathrm{v}\\
            \mathrm{u}^Tm(M_1,\ldots,M_n) & \mathrm{u}^Tm(M_1,\ldots,M_n)\mathrm{v}
        \end{bmatrix}\,.
    \end{align}
\end{proposition}
\begin{proof}
    Let $\cal M$ be the set of matrices that are neither first nor last in the monomial $m$. Without loss of generality, let $M_1$ and $M_n$ be the first and last matrices, respectively, in the monomial $m$. 
    Then we write $m(M_1,\ldots,M_n) = M_1 m_{\text{middle}}({\cal M}) M_n$, where $m_{\text{middle}}$ is the monomial without the first and last matrices. Note that, for any matrix $M\in {\cal M}$, $M$ is modified into $M^{(\textrm{middle})}$. 
    Let $\hat{\cal M}$ be the set of such matrices and let $\hat{m}_{\text{middle}}$ be the monomial with the same multiplication structure as $m_{\text{middle}}$ but with the modified inputs.
    In other words, we have $\hat{m}(\hat{M}_1,\ldots,\hat{M}_{\hat{n}}) = M_1^{(\text{first})}\hat{m}_{\text{middle}}(\hat{\cal M}) M_n^{(\text{last})}$ Then
    \begin{align*}
        \hat{m}&(\hat{M}_1,\ldots,\hat{M}_{\hat{n}})\\ &= \begin{bmatrix}
            M_1 & \pmb{0}\\
            \mathrm{u}^TM_1 & 0
        \end{bmatrix}\begin{bmatrix}
         m_{\text{middle}}({\cal M}) & \pmb{0}\\
         \pmb{0}^T & 0
         \end{bmatrix}\begin{bmatrix}
             M_n & M_n\mathrm{v}\\
             \pmb{0}^T & 0
         \end{bmatrix}\\
         &= \begin{bmatrix}
             M_1m_{\text{middle}}({\cal M}) & \pmb{0}\\
             \mathrm{u}^TM_1m_{\text{middle}}({\cal M}) & 0
         \end{bmatrix}\begin{bmatrix}
             M_n & M_n\mathrm{v}\\
             \pmb{0}^T & 0
         \end{bmatrix}\\
         &= \begin{bmatrix}
            m(M_1,\ldots,M_n) & m(M_1,\ldots,M_n)\mathrm{v}\\
            \mathrm{u}^Tm(M_1,\ldots,M_n) & \mathrm{u}^Tm(M_1,\ldots,M_n)\mathrm{v}
        \end{bmatrix}\,.
    \end{align*}
\end{proof}

\begin{proposition}
\label{thm:polywithmodif}
    Consider a polynomial $\psi$ with $s$ matrices $M_1,\ldots,M_s$ as inputs. 
    Let $m_1,\ldots, m_k$ be the monomials of $\psi$ and suppose that the monomial $m_i$ has $s_i$ inputs for $i\in[k]$. 
    That is,
    \begin{align}
        \psi(M_1,\ldots,M_s)& = m_1\left(M_{(1)}^{(1)},\ldots, M_{(1)}^{(s_1)}\right) \notag\\
        &\hspace{2mm}+ \cdots + m_k\left(M_{(k)}^{(1)},\ldots, M_{(k)}^{(s_k)}\right).
    \end{align}
    Let $\hat{M}_1,\ldots,\hat{M}_{\hat{s}}$ be the modified matrices with nonzero random vectors $\mathrm{u},\mathrm{v}\in\FF^n$ as proposed and let $\hat\psi$ be the modified polynomial, with $\hat{s}$ inputs, which consists of modified monomials $\hat{m}_i$'s with the same multiplication structure as $m_i$'s, but with $\hat{s}_i$ inputs. Then,
    \begin{align}
        &\hat{\psi}(\hat{M}_1,\ldots,\hat{M}_{\hat{s}})\notag\\ &= \begin{bmatrix}
            \psi(M_1,\ldots,M_s) & \psi(M_1,\ldots,M_s)\mathrm{v}\\
            \mathrm{u}^T\psi(M_1,\ldots,M_s)  & \mathrm{u}^T\psi(M_1,\ldots,M_s) \mathrm{v}
        \end{bmatrix}.
    \end{align}
    \begin{proof}
        Applying Proposition~\ref{thm:mono_with_mod_inp}, we have
        {\small
        \begin{align*}
            &\hat{\psi}(\hat{M}_1,\ldots,\hat{M}_{\hat{s}}) \\&=  \sum_{i=1}^k \hat{m}_i \left(\hat{M}_{(i)}^{(1)},\ldots, \hat{M}_{(i)}^{(\hat{s}_i)}\right)\\
            &= \sum_{i=1}^k\begin{bmatrix} m_i(M_{(i)}^{(1)},\ldots,M_{(i)}^{(s_i)}) & m_i(M_{(i)}^{(1)},\ldots,M_{(i)}^{(s_i)})\mathrm{v}\\
                \mathrm{u}^Tm_i(M_{(i)}^{(1)},\ldots,M_{(i)}^{(s_i)}) & \mathrm{u}^Tm_k(M_{(i)}^{(1)},\ldots,M_{(i)}^{(s_i)})\mathrm{v}
            \end{bmatrix}
            \\&=\begin{bmatrix}
            \psi(M_1,\ldots,M_s) & \psi(M_1,\ldots,M_s)\mathrm{v}\\
            \mathrm{u}^T\psi(M_1,\ldots,M_s)  & \mathrm{u}^T\psi(M_1,\ldots,M_s) \mathrm{v}
        \end{bmatrix}.
        \end{align*}
        }
    \end{proof}
\end{proposition}
Proposition~\ref{thm:polywithmodif} shows that, if the Master Node performs distributed computations using the modified inputs (with the modified polynomials), the resulting matrix contains the required computations and some other elements that can be used for verification purposes. Suppose that one of the modified computations results in $\begin{bmatrix}
    M & M\mathrm{v}\\
    \mathrm{u}^TM & \mathrm{u}^TM\mathrm{v}
\end{bmatrix}$. Then $M$ is one of the required computations and the Master Node performs the following checks:
\begin{enumerate}
    \item compute $M\mathrm{v}$ and check if $M\mathrm{v}$ is equal to the first $n$ components in the last column,
    \item compute $\mathrm{u}^T M$ and check if $\mathrm{u}^TM$ is equal to the first $n$ components in the last row, and
    \item compute $\mathrm{u}^T M\mathrm{v}$ and check if $\mathrm{u}^TM\mathrm{v}$ is equal to the bottom right component.
\end{enumerate}
The theorem below shows that this verification scheme can detect the presence of incorrect computations with high probability.
\begin{theorem}\label{thm::verificationB}
    Scheme 4B can detect the presence of incorrect computations in the presence of up to $T$ adversarial nodes with probability $1-O\left(\frac{1}{q}\right)$.
\end{theorem}
\begin{proof}
    Suppose that the recovered computations (with modified inputs) are of the form $\begin{bmatrix}
        M_{i,j} & x_{i,j}\\
        y_{i,j}^T & z_{i,j}
    \end{bmatrix}$, for all $i\in[\ell_j], j\in[L]$, where $M_{i,j}$ is the original computation $\psi_j\left(\pmb{X}_i^{(j)}\right)$. Let $\begin{bmatrix}
        \tilde{M}_{i,j} & \tilde{x}_{i,j}\\
        \tilde{y}_{i,j}^T & \tilde{z}_{i,j}
    \end{bmatrix}$ be the matrix obtained from Worker responses in the presence of at most $T$ adversarial workers. Let $E$ be the event when the adversaries successfully persuade the Master Node to accept the wrong results. 
    That is,
    \begin{align}
        E = \bigcup_{i\in[\ell_j], j\in[L]} (A_{i,j} \cap B_{i,j} \cap C_{i,j} \cap D_{i,j}),
    \end{align}
    where $A_{i,j}$ is the event when $M_{i,j} \ne \tilde{M}_{i,j}$, $B_{i,j}$ is the event when $\tilde{x}_{i,j} = \tilde{M}_{i,j}\mathrm{v}_j$, $C_{i,j}$ is the event when $\tilde{y}_{i,j}^T = \mathrm{u}_j^T \tilde{M}_{i,j}$ and $D_{i,j}$ is the event when $\tilde{z}_{i,j} = \mathrm{u}_j^T \tilde{M}_{i,j}\mathrm{v}$. 
    Then
    \begin{align}
        \mathbb{P}(E)
        &= \mathbb{P}\left(\bigcup_{i\in[\ell_j], j\in[L]} (A_{i,j} \cap B_{i,j} \cap C_{i,j} \cap D_{i,j})\right)\notag\\
        &= \sum_{i\in[\ell_j],j\in[L]}\mathbb{P}\left(A_{i,j} \cap B_{i,j} \cap C_{i,j} \cap D_{i,j} \right)\notag\\
        &\le  \sum_{i\in[\ell_j],j\in[L]}\min\Bigg\{\mathbb{P}(A_{i,j}\cap B_{i,j}),\mathbb{P}(A_{i,j}\cap C_{i,j}),\notag\\
        & \hspace{3cm} \mathbb{P}(A_{i,j}\cap D_{i,j})\Bigg\}\,.
    \end{align}
    
The condition $A_{i,j}\cap B_{i,j}$ is equivalent to
\begin{align}
    0 &= \tilde{x}_{i,j} - \tilde{M}_{i,j}\mathrm{v}_j\notag \\
    &=  \tilde{x}_{i,j} - x_{i,j} - (\tilde{M}_{i,j} - M_{i,j})\mathrm{v}_j,
\end{align}
where $\tilde{M}_{i,j}-M_{i,j}$ is nonzero. The terms $\tilde{x}_{i,j} - x_{i,j}$ and $\tilde{M}_{i,j} - M_{i,j}\ne 0$ are controlled by adversary nodes (hence, these terms are deterministic and independent to $\mathrm{v}$). Therefore, the event $A_{i,j}\cap B_{i,j}$ is equivalent to $G^{(A,B)}_{i,j}(\mathrm{v}_j) = 0$, where $G^{(A,B)}_{i,j}(\mathrm{v}) = a_{i,j} - B_{i,j}\mathrm{v}$, $B_{i,j}\ne 0$. Hence, by Schwartz-Zippel Lemma~\cite{SZ1, SZ2}, $\mathbb{P}(A_{i,j}\cap B_{i,j}) = O\left(\frac{1}{q}\right)$. By the same technique we can easily get that $\mathbb{P}(A_{i,j}\cap C_{i,j}) =\mathbb{P}(A_{i,j}\cap D_{i,j}) =O\left(\frac{1}{q}\right)$.
\end{proof}

\section{Numerical Results}\label{sec::numerical}

In this section, we discuss the performance of our distributed computing schemes. First, we look at the performance of the schemes without verification (Scheme 1, 2, 3, and 4) in solving the problem in Example~\ref{example::scheme3}, which is to compute $A^2, B^2$ and $AB+BC^2$ in the presence of one colluding worker and one adversarial worker, without verification. Then, we compare the performance of our verification schemes which are built on top of Scheme 4. The Straggler Resistance (SR), Upload Cost (UC), and Download Cost (DC) for each scheme can be found in Table~\ref{tab:scheme_no_verif_1}, and the computation costs can be found in Table~\ref{tab:scheme_no_verif_2}.
\begin{table}[]
    \captionsetup{font=small}
    \centering
    \fontsize{8}{11}\selectfont
    \begin{tabularx}{0.5\textwidth}{ @{}|Y|Y|Y|Y|@{} }
    \hline  
       Method  & SR & UC & DC \\
       \hline
       Scheme 1  & $d-12$ & $4dn^2$ & $12n^2$ \\
      \hline
       Scheme 2  & $\left\lfloor\dfrac{d-13}{2}\right\rfloor$ & \makecell{$(d_1+3d_2)n^2$\\$d_1+d_2 = d$} & $13n^2$ \\
       \hline
       Scheme 3  & $d-7$ & $4dn^2$ & $13n^2$ \\
       \hline
       Scheme 4  & $d-8$ & $4dn^2$ &  $8n^2$ \\
       \hline
       Scheme 4A & $d-8$ & $9dn^2$ & $16n^2$\\
       \hline
       Scheme 4B & $d-8$ & $7d(n+1)^2$ & $8(n+1)^2$\\
       \hline
    \end{tabularx}
    \caption{The Straggler Resistance (SR), Upload Cost (UC), and Download Cost (DC) to obtain $A^{2}, B^{2}, AB+BC^2$ as in Example~\ref{example::scheme3}, with and without verification. We express UC and DC in terms of the number of field elements $\FF$.}
    \label{tab:scheme_no_verif_1}
\end{table}

\begin{table}[]
    \captionsetup{font=small}
    \centering
    \fontsize{8}{11}\selectfont
    \begin{tabularx}{0.5\textwidth}{ @{}|Y|Y|Y|Y|Y|@{} }
       \hline
       \multirow{2}{*}{Method}  & \multicolumn{2}{c|}{Computation on MN} & \multicolumn{2}{c|}{Computation on WN} \\
       \cline{2-5}
       & $\cdot$ & $\times$ & $\cdot$ & $\times$\\
       \hline
       Scheme 1  & $16$ & $0$ & $0$  & $4$ ($n$ by $n$)\\
       \hline
       Scheme 2  & \makecell{$G_1:3$ \\ $G_2:6$} & $0$ & $0$ & \makecell{$G_1:1$\\ $G_2:3$\\ ($n$ by $n$)} \\
       \hline
       Scheme 3  & $9$ & $0$ & $0$ & $4$ ($n$ by $n$)\\
       \hline
       Scheme 4  & $9$ & $0$ &  $2$ & $4$ ($n$ by $n$)\\
       \hline
       Scheme 4A & $20$ & $0$ & $4$ & $8$ ($n$ by $n$)\\
       \hline
       Scheme 4B & $16$ & $0$ & $2$ & \makecell{$4$ ($n+1$\\ by $n+1$)}\\
       \hline
    \end{tabularx}
    \caption{The computation costs to obtain $A^{2},B^{2},AB+BC^2$ as in Example~\ref{example::scheme3}.}
    \label{tab:scheme_no_verif_2}
\end{table}

We can see that, in terms of Straggler Resistance, Scheme~2 performs the worst. Scheme~4 performs better than Scheme~1 but slightly worse than Scheme~3. All schemes have comparable Upload Costs but the Download Cost of Scheme~4 is lower than other schemes. Furthermore, Master Nodes don't perform any matrix-matrix multiplication $\times$ in all schemes but only scalar multiplications, which are easy to perform. However, Scheme~4 slightly increases the workload of the Worker Node by $2$ extra scalar multiplications.

Now, let us compare Scheme 4 with Scheme 4A and 4B. Firstly, they all have the same Straggler Resistance as the verification schemes do not require extra responding nodes to perform. However, our verification schemes increase communication costs. The Upload Cost increases from $4dn^2$ to $9dn^2$ for Scheme 4A and $7d(n+1)^2$ for Scheme 4B. On the other hand, the Download Cost increases from $8n^2$ to $16n^2$ for Scheme 4A and $8(n+1)^2$ for Scheme 4B. For big $n$, Scheme 4B has lower communication costs. 

The computation on Master Node increases from $9$ scalar multiplications to $20$ for Scheme 4A  and $16$ for Scheme 4B. On Worker Nodes, Scheme 4A doubles the computation. However, with Scheme 4B, the Worker Nodes only need to do the same number of multiplications of $n+1$ by $n+1$ matrices. For big $n$, Scheme $4B$ has lower computation costs.
\section{Conclusion}\label{sec::conclusion}

We considered the problem of efficiently evaluating arbitrary multivariate polynomials over several massive datasets in a distributed computing system. We proposed a new scheme based on the Lagrange Coded Computing framework and compared its efficiency against several naive schemes that provide a solution to a problem of our kind. While our proposed scheme has slightly worse straggler resistance in comparison to some of the naive schemes to the problem of our kind, we observe that it provides a significantly lower download cost in comparison to all competing schemes. On top of it, we propose two verification schemes to detect the existence of wrong results without increasing the number of required responses. One scheme works for a general distributed computing set-up, while another one only works for square matrices. However, the latter has lower communication and computation costs than the former. Generalizing proposed approaches for machine-learning functions and real-number cases are interesting open problems.

{\section*{Acknowledgements.}
This research/project is supported by the National Research Foundation, Singapore under its Strategic Capability Research Centres Funding Initiative, Singapore Ministry of Education Academic Research Fund Tier 2 Grants MOE2019-T2-2-083 and MOE-T2EP20121-0007. Any opinions, findings and conclusions or recommendations expressed in this material are those of the author(s) and do not reflect the views of National Research Foundation, Singapore.}

\balance
\bibliographystyle{IEEEtran}
\bibliography{OneRoundCodedComputingofMultipleFunctions}

\end{document}